\documentclass[11pt]{article}
\usepackage{amssymb}
\usepackage{amsmath}
\usepackage{amsthm}
\usepackage[dvips]{graphics,color}
\usepackage{epsfig}
\usepackage{a4wide}
\def\E{\mathbb{E}}
\def\var{\text{Var}}
\def\cov{\text{Cov}}

\def\N{\mathcal{N}}
\def\W{\mathcal{W}}
\def\IW{\mathcal{IW}}
\def\B{\mathcal{B}}

\def\T{\mathcal{T}}
\newtheorem{thm}{Theorem}
\newtheorem{lem}{Lemma}

\newtheorem{pro}{Proposition}

\begin{document}

\title{Multivariate Stochastic Volatility with Bayesian Dynamic
Linear Models}
\author{K. Triantafyllopoulos\footnote{Department of Probability and Statistics,
Hicks Building, University of Sheffield, Sheffield S3 7RH, UK,
Email: \tt{k.triantafyllopoulos@sheffield.ac.uk}}}

\date{\today}

\maketitle

\begin{abstract}

This paper develops a Bayesian procedure for estimation and
forecasting of the volatility of multivariate time series. The
foundation of this work is the matrix-variate dynamic linear model,
for the volatility of which we adopt a multiplicative stochastic
evolution, using Wishart and singular multivariate beta
distributions. A diagonal matrix of discount factors is employed in
order to discount the variances element by element and therefore
allowing a flexible and pragmatic variance modelling approach.
Diagnostic tests and sequential model monitoring are discussed in
some detail. The proposed estimation theory is applied to a
four-dimensional time series, comprising spot prices of aluminium,
copper, lead and zinc of the London metal exchange. The empirical
findings suggest that the proposed Bayesian procedure can be
effectively applied to financial data, overcoming many of the
disadvantages of existing volatility models.

\textit{Some key words:} Time series, volatility, multivariate,
dynamic linear model, Bayesian, forecasting, state space, Kalman
filter, GARCH, London metal exchange.

\end{abstract}

\section{Introduction}\label{intro}

In the last two decades, multivariate time series have received
considerable attention with the emphasis being placed on state space
models (L\"utkepohl, 1993, West and Harrison, 1997, Chapter 16;
Durbin and Koopman, 2001, Chapter 3; De Gooijer and Hyndman, 2006).
From an econometrics standpoint time-varying volatility models have
been widely developed, recognizing the essence that the volatility
and the correlation of assets change over time. Although univariate
volatility models are useful in estimating and forecasting
volatility, it is widely recognized (Bauwens {\it et al.}, 2006)
that multivariate models, which can model the serial and cross
correlation of the assets, should be employed.

From a time series standpoint, volatility models are developed
within two main families of models: the multivariate generalized
autoregressive conditional heteroskedastic (MGARCH), including the
multivariate ARCH, and the multivariate stochastic volatility (MSV)
families. Multivariate ARCH models include the diagonal vech model
(Bollerslev {\it et al.}, 1988), the constant conditional
correlation model (Bollerslev, 1990), the factor-ARCH model (Engle
{\it et al.}, 1990), the BEKK model (Engle and Kroner, 1995) and the
latent factor ARCH model (Diebold and Nerlove, 1989); see also Wong
and Li (1997), Tse and Tsui (2002), Comte and Lieberman (2003), and
Audrino and Barone-Adesi (2006). MSV models have also received a lot
of attention, see e.g. Harvey {\it et al.} (1994), Jacquier {\it et
al.} (1995), Kim {\it et al.} (1998), Pitt and Shephard (1999),
Aguiliar and West (2000) and Meyer {\it et al.} (2003). A number of
estimation procedures have been suggested for MSV models; for
instance, see Bauwens {\it et al.} (2006), Yu and Meyer (2006),
Liesenfeld and Richard (2006), Asai {\it et al.} (2006) and Maasoumi
and McAleer (2006). In this context, several variations of
computationally expensive Markov chain Monte Carlo (MCMC) methods
are commonly used following papers by Shephard (1993), Jacquier {\it
et al.} (1994), Kim {\it et al.} (1998), Shephard and Pitt (1997),
Uhlig (1997), Chib {\it et al.} (2002) and Philipov and Glickman
(2006a, 2006b).

Most of the proposed models are aimed at specific applications, or
they impose restrictions in the parameter space, or they are only
available for data with low dimensionality. In particular, it would
be desirable to obtain estimation algorithms, for which the model
would estimate not only the volatility covariance matrix, but also
shocks in the levels of the returns. In addition to that, it is
desirable to construct a model that will not rely on Monte Carlo or
any other simulation procedures and also will not target data of
specific applications.

In this paper we develop a general state space model, which allows
the volatility covariance matrix to be estimated with a fast
Bayesian algorithm. The proposed algorithm is achieved by
considering a stochastic multiplicative model for the volatility,
which is based on Wishart and singular multivariate beta
distributions. A diagonal matrix of degrees of freedom is used in a
variance discounting approach in order to update the estimates and
the forecasts of the volatility from time $t-1$ to time $t$. This
has a unique advantage that different volatilities can be discounted
at different rates, for example one can have two assets, the
volatility of the first changes at a rate according to a discount
factor of $0.7$ and the volatility of the second changes at a slower
rate according to a discount factor of $0.95$. The algorithm is fast
and provides not only one-step ahead forecasts of the volatility,
but also the entire one-step ahead forecast distribution of the
volatility. A Bayesian algorithm is outlined for sequential model
comparison. The proposed methodology is illustrated by considering
data, consisting of spot prices of aluminium, copper, lead and zinc
from the London metal exchange. It is found that the volatilities of
aluminium and zinc prices are driven from a common factor and the
volatilities of copper and lead prices are driven from another
factor, while the respective correlations are around $\pm 0.5$. The
performance of the model is discussed by using several diagnostic
toolkits, including the mean of squared standardized forecast
errors, the log-likelihood function and Value-at-Risk.

The paper is organized as follows. Section \ref{model:dwr} defines
the model, for which inference is developed in Section
\ref{estimation}. Section \ref{test} discusses diagnostic tests and
model comparison, and the following section analyzes data from the
London metal exchange market. In Section \ref{conclusions} we
discuss the advantages of the proposed approach as compared with
existing GARCH estimation procedures. The appendix gives full
mathematical details (including the proofs) of arguments in Sections
\ref{estimation} and \ref{test}.

\section{Matrix-Variate Dynamic Linear Models}\label{model:dwr}

Matrix-variate dynamic linear models (MV-DLMs) are introduced in
Quintana and West (1987) and they are further developed in Salvador
{\it et al.} (2003), Salvador and Gargallo (2004), Salvador {\it et
al.} (2004), Triantafyllopoulos and Pikoulas (2002) and
Triantafyllopoulos (2006a); matrix-variate DLMs are reported in some
detail in West and Harrison (1997, \S16.4). For the purpose of this
paper the discussion is restricted to vector-valued time series; the
general description for matrix-valued time series can be found in
Salvador {\it et al.} (2003). We should note that from a frequentist
standpoint, MV-DLMs have been developed in Harvey (1986, 1989),
Harvey and Snyder (1990), and Fern\'{a}ndez and Harvey (1990).
Suppose that the $p$-dimensional response vector $y_t$ follows a
matrix-variate DLM so that
\begin{equation}\label{model2}
y_t'=F_t'\Theta_t+\epsilon_t' \quad \textrm{and} \quad \Theta_t =
G_t \Theta_{t-1}+\omega_t,
\end{equation}
where $F_t$ is a $d$-dimensional design vector, $G_t$ is a $d\times
d$ evolution matrix and $\Theta_t$ is a $d\times p$ state matrix.
Conditional on $\Sigma_t$, the innovations $\epsilon_t$ and
$\omega_t$ follow, respectively, multivariate and matrix-variate
normal distributions, i.e.
$$
\epsilon_t|\Sigma_t\sim \N_{p\times 1}(0,\Sigma_t)\quad \textrm{and}
\quad \omega_t|\Sigma_t\sim \N_{d\times p}(0,\Omega_t,\Sigma_t),
$$
where $\Sigma_t$ is the unknown $p\times p$ volatility covariance
matrix of the innovations $\epsilon_t$, and $\Omega_t$ is a $d\times
d$ covariance matrix of the innovation $\omega_t$. The distribution
of $\omega_t|\Sigma_t$ may also be written as
\begin{displaymath}
\textrm{vec}(\omega_t)|\Sigma_t\sim \N_{dp\times
1}(0,\Sigma_t\otimes \Omega_t),
\end{displaymath}
where $\textrm{vec}(\cdot)$ denotes the column stacking operator of
a matrix and $\otimes$ denotes the Kronecker product. It is assumed
that the innovation series $\{\epsilon_t\}$ and $\{\omega_t\}$ are
internally and mutually uncorrelated and also they are uncorrelated
with the assumed priors
\begin{equation}\label{dwr3}
\Theta_0|\Sigma_0\sim \N_{d\times
p}(m_0,P_0,\Sigma_0)\quad\textrm{and}\quad \Sigma_0 \sim
\IW_p(n_0+2p,S_0),
\end{equation}
for some known $m_0$, $P_0$, $n_0$ and $S_0$. Here
$\Sigma\sim\IW_p(k,S)$ denotes the inverted Wishart distribution
with $k$ degrees of freedom and parameter matrix $S$ with density
function
$$
p(\Sigma) = \frac{2^{-(k-p-1)p/2}
|S|^{(k-p-1)/2}}{\Gamma_p\{(k-p-1)/2\} |\Sigma|^{k/2}}
\textrm{etr}\left(-\frac{1}{2}S\Sigma^{-1}\right), \quad k>2p,
$$
where $\Gamma_p(\cdot)$ denotes the multivariate gamma function,
$\textrm{etr}(\cdot)$ denotes the exponent of a trace of a matrix,
and $|S|$ denotes the determinant of $S$. Then $\Sigma^{-1}$ follows
the Wishart distribution $\W_p(k-p-1,S^{-1})$. Let $N$ be a positive
integer and write $y^t=\{y_1,y_2,\ldots,y_t\}$ the information set
comprising observations up to time $t$, for $t=1,2,\ldots,N$. The
covariance matrix $\Omega_t$ is specified with at most $d$ discount
factors $\delta_1,\delta_2,\ldots,\delta_d$ so that
$$
\Sigma_{t-1}\otimes\Omega_t=\var\left\{\textrm{vec}\left(\Delta^{1/2}G_t
\Theta_{t-1}|\Sigma_{t-1},y^{t-1}\right) \right\},
$$
where
$\Delta=\textrm{diag}\{(1-\delta_1)/\delta_1,\ldots,(1-\delta_d)/\delta_d\}$.
Thus $\Omega_t$ is the implied covariance matrix obtained after
discounting is used in order to increase the covariance matrix from
time $t-1$ to time $t$, given information $y^{t-1}$. The above
equation justifies that
\begin{eqnarray*}
\Sigma_{t-1}\otimes\Omega_t&=&\var\{(I_p\otimes
\Delta^{1/2}G_t)\textrm{vec}(\Theta_{t-1})|\Sigma_{t-1}\} =
(I_p\otimes \Delta^{1/2}G_t)(\Sigma_{t-1}\otimes P_{t-1})(I_p\otimes
G_t'\Delta^{1/2}) \\ &=& \Sigma_{t-1} \otimes
\Delta^{1/2}G_tP_{t-1}G_t'\Delta^{1/2},
\end{eqnarray*}
implying $\Omega_t=\Delta^{1/2}G_tP_{t-1}G_t'\Delta^{1/2}$, where
$P_{t-1}$ is the left covariance matrix of $\Theta_{t-1}|y^{t-1}$,
so that $\Sigma_{t-1}\otimes
P_{t-1}=\var\{\textrm{vec}(\Theta_{t-1})|\Sigma_{t-1},y^{t-1}\}$ (in
Section \ref{estimation} it is shown that $P_{t-1}$ is calculated
routinely). It is proposed that the above setting for $\Omega_t$ is
carried out for the covariance matrix
$\var\{\textrm{vec}(\omega_t)|\Sigma_t\}=\Sigma_t\otimes \Omega_t$.
This setting, which generalizes the single discounting approach of
West and Harrison (1997), is necessary to consider in order to
retain conjugate forms in the updating of the posterior distribution
of $\Theta_t|\Sigma_t,y^t$ (see Section \ref{estimation}). Multiple
discount factors are useful in capturing the different structural
characteristics of trend, seasonal and regression coefficient
elements of the evolution matrix $\Omega_t$.

The volatility matrix $\Sigma_t$ imposes complications in inference,
but, it is a very useful consideration in the model because in
financial time series, high frequency data exhibit short-term or
long-term heteroscedastic behaviour. In the remainder of this
section we describe the stochastic model governing the evolution of
$\Sigma_t$.

At time $t-1$ we assume that $\Sigma_{t-1}$, conditional on
$y^{t-1}$, follows an inverted Wishart distribution, $\Sigma_{t-1} |
y^{t-1}\sim \IW_p(n+2p, S_{t-1})$, for some $n$ and $S_{t-1}$. The
precision matrix is indicated by $\Phi_t = \Sigma_t^{-1}$ and,
following a Choleski decomposition, we write $\Phi_t = C_t'C_t$,
where $C_t$ is the unique upper triangular matrix of the Choleski
decomposition. The law governing the evolution of the $\Sigma_t$ or
$\Phi_t$ from time $t-1$ to time $t$ is represented by
\begin{equation}\label{vol:evol}
\Phi_t = \beta^{-1/2} C_{t-1}' B_t C_{t-1} \beta^{-1/2},
\end{equation}
where $\beta=\textrm{diag}(\beta_1,\beta_2,\ldots,\beta_p)$ is
diagonal matrix of discount factors $\beta_1,\beta_2,\ldots,\beta_p$
and $B_t$, which given $y^{t-1}$, is assumed to be independent of
$\Phi_{t-1}$, is a random matrix following the singular multivariate
beta distribution with parameters $(p^{-1}\textrm{tr}(\beta)
n+p-1)/2$ and $1/2$; we write $B_t | y^{t-1}\sim
\B_p\{(p^{-1}\textrm{tr}(\beta) n+p-1)/2,1/2\}$. In Section
\ref{estimation} we will see that
$n=1/(1-p^{-1}\textrm{tr}(\beta))$. Of course $n$ is defined for
$0<\beta_i<1$, for $i=1,\ldots,p$. For more details on the singular
multivariate beta distribution the reader is referred to Uhlig
(1994), D\'{i}az-Garc\'{i}a and Guti\'{e}rrez (1997), and Srivastava
(2003).

The evolution (\ref{vol:evol}) is motivated from the univariate case
($p=1$), for which (\ref{vol:evol}) reduces to
\begin{equation}\label{evoluni}
\Phi_t=\beta^{-1}\Phi_{t-1}B_t.
\end{equation}
In this case the multivariate singular beta reduces to a standard
beta distribution and as $B_t$ is independent of $\Phi_{t-1}$, we
have $\E(\Phi_t|y^{t-1})=\E(\Phi_{t-1}|y^{t-1})$ and
$\var(\Phi_t|y^{t-1})> \var(\Phi_{t-1}|y^{t-1})$, since $0<\beta<
1$. This defines a random walk type evolution for $\Phi_t$. The
above evolution for a scalar volatility $\Sigma_t$ is studied in
Harrison and West (1987), West and Harrison (1997, \S10.8), and
Triantafyllopoulos (2007).

Returning to the case when $p\geq 1$, suppose that
$\beta_1=\cdots=\beta_p$ and so $\beta=\beta_1I_p$. In this case the
evolution (\ref{vol:evol}) reduces to
$$
\Phi_t=\beta_1^{-1}C_{t-1}'B_tC_{t-1},
$$
where $0<\beta_1<1$. In Proposition \ref{pro1} of the appendix, it
is shown that $\Sigma_t|y^{t-1}\sim
\IW_p(p^{-1}\textrm{tr}(\beta)n+2p, \beta^{1/2} S_{t-1}
\beta^{1/2})$ or $\Phi_t|y^{t-1}\sim\W_p(\beta_1n+p-1,
\beta_1^{-1}S_{t-1}^{-1})$ and so
\begin{equation}\label{evolsmall}
\E(\Phi_t|y^{t-1})=(\beta_1
n+p-1)\beta_1^{-1}S_{t-1}^{-1}=\left(n+\frac{p-1}{\beta_1}\right)S_{t-1}^{-1},
\end{equation}
which is different than
$\E(\Phi_{t-1}|y^{t-1})=(n+p-1)S_{t-1}^{-1}$, unless $\beta_1=1$.
It follows that the random walk type evolution of (\ref{evoluni})
is retained for values of $\beta_1$ close to 1, but otherwise the
evolution (\ref{vol:evol}) defines a shrinkage type evolution, for
which $\E(\Phi_t|y^{t-1})>\E(\Phi_{t-1}|y^{t-1})$. Our empirical
results of Section \ref{analysis} show that the estimator of
$\Sigma_t$, which is generated from evolution (\ref{vol:evol}),
performs well for relatively high values of the discount matrix
$\beta$. For $p=1$, West and Harrison (1997, \S10.8) suggest a
slow evolution (\ref{evoluni}), for which a discount factor close
to 1 is proposed. In particular on page 361 of the above reference
it is stated ``We note that practically suitable variance discount
factors take values near unity, typically between 0.95 and 0.99''.
This is in agreement with our proposal, in the general case of
$p\geq 1$, so that the shrinkage effect in (\ref{vol:evol}) is
small. However, our empirical results in Section \ref{analysis}
suggest that the modeller should allow for smaller values of the
discount factors in the range of 0.6 and 0.99 so that shocks in
the volatility can be estimated. The evolution (\ref{evolsmall})
makes the assumption that all elements of $\Phi_t$ are discounted
at the same rate via the single discount factor $\beta_1$.
Equation (\ref{vol:evol}) introduces a flexible evolution, where
each of the diagonal elements of $\Phi_t$ are discounted at
different rate via the $p$ discount factors
$\beta_1,\ldots,\beta_p$.

\section{Estimation} \label{estimation}

From evolution (\ref{vol:evol}) and Proposition \ref{pro1} of the
appendix, the prior density of $\Sigma_t| y^{t-1}$ is the inverted
Wishart density
\begin{equation}\label{prior:v}
\Sigma_t | y^{t-1} \sim \IW_p(p^{-1}\textrm{tr}(\beta)
n+2p,\beta^{1/2} S_{t-1}\beta^{1/2}),
\end{equation}
where $n=1/(1-p^{-1}\textrm{tr}(\beta))$.

Without loss in clarity of the presentation, we denote by $p(X)$
the probability density function of a random matrix $X$, avoiding
to explicitly write $p_X(X)$. Thus if $X$ and $Y$ denote two
different random matrices, $p(X)$ and $p(Y)$ denote respectively
the densities of $X$ and $Y$.

From model (\ref{model2}), given $\Sigma_t$ and $y^{t-1}$, the joint
distribution of $y_t'$ and $\Theta_t$ is
\begin{equation}\label{eq:joint}
\left[\left.\begin{array}{c} y_t' \\ \Theta_{t} \end{array} \right|
\Sigma_t, y^{t-1}\right] \sim \N_{(d+1)\times p}\left( \left[
\begin{array}{c} F_t'G_tm_{t-1} \\ G_tm_{t-1} \end{array}\right], \left[
\begin{array}{cc} F_t'R_{t}F_t & F_t'R_{t} \\ R_{t} F_t & R_{t}\end{array}\right],
\Sigma_{t} \right),
\end{equation}
where $R_{t}=G_tP_{t-1}G_t'+\Omega_t$ and the covariance of $y_{t}'$
and $\Theta_{t}$ is determined by
\begin{eqnarray*}
\cov\{\textrm{vec}(y_t'),\textrm{vec}(\Theta_t) |
\Sigma_t,y^{t-1}\}&=&
\cov\{\textrm{vec}(F_t'\Theta_{t}+\epsilon_{t}'),\textrm{vec}(\Theta_{t})|\Sigma_{t},y^{t-1}\}
\\ &=& \cov\{(1\otimes
F_t')\textrm{vec}(\Theta_{t}),\textrm{vec}(\Theta_{t})|\Sigma_{t},y^{t-1}\}
\\ &=& (1\otimes F_t')
\var\{\textrm{vec}(\Theta_{t})|\Sigma_{t},y^{t-1}\} \\
&=&(1\otimes F_t')(\Sigma_{t}\otimes R_{t}) = \Sigma_{t} \otimes
(F_t' R_t).
\end{eqnarray*}

From (\ref{eq:joint}) and the inverted Wishart prior (\ref{prior:v})
it follows that the joint forecast density of $y_{t}'$ and
$\Theta_{t}$, given only $y^{t-1}$ is a $p\times 1$ multivariate
Student $t$ density (see Theorem 4.2.1 of Gupta and Nagar, 1999,
\S4.2), i.e.
\begin{eqnarray*}
\left[\left.\begin{array}{c} y_t' \\ \Theta_{t} \end{array} \right|
y^{t-1}\right] & \sim & \T_{p\times 1}\bigg(p^{-1}\textrm{tr}(\beta)
n+p-1, \left[
\begin{array}{c} F_t'G_tm_{t-1} \\ G_tm_{t-1} \end{array}\right], \left[
\begin{array}{cc} F_t'R_{t}F_t & F_t'R_{t} \\ R_{t} F_t &
R_{t}\end{array}\right], \\ &&  \beta^{1/2} S_{t-1} \beta^{1/2}
\bigg)  \equiv  \T_{p\times 1}( \nu, M, U, S),
\end{eqnarray*}
with density
\begin{eqnarray*}
p(T|y^{t-1})&=&
\frac{\Gamma_p\{(\nu+d+p-1)/2\}}{\pi^{dp/2}\Gamma_p\{(\nu+p-1)/2\}}
|U|^{-p/2} |S|^{(\nu+p-1)/2} \\ &&\times
|S+(T-M)'U^{-1}(T-M)|^{-(\nu+d+p-1)/2},
\end{eqnarray*}
where $T'=[y_{t} ~ \Theta_{t}']$.

Applying Bayes' theorem, the posterior distribution of $\Sigma_{t} |
y^{t}$ results to be an inverted Wishart. To detail the derivations
of this result we need to note that the likelihood function of
$\Sigma_{t}$ from the single observation $y_{t}$ is
$L(\Sigma_{t};y_{t})=p(y_{t}'|\Sigma_{t},y^{t-1})$, whilst the prior
of $\Sigma_{t}$ is given by (\ref{prior:v}). Thus the posterior of
$\Sigma_{t}$ given $y^{t}$ is
\begin{eqnarray*}
p(\Sigma_{t}| y^{t}) &=& \frac{L(\Sigma_{t};y_{t}) p(\Sigma_{t}|
y^{t-1})}{p(y_{t}'| y^{t-1})}
=\frac{p(y_t|\Sigma_{t},y^{t-1}) p(\Sigma_{t}| y^{t-1})}{p(y_{t}'| y^{t-1})} \\
&\propto& |\Sigma_{t}|^{-1/2}\textrm{etr} \left\{-\frac{1}{2}
(y_{t}'-F_t'G_tm_{t-1})Q_{t}^{-1}(y_{t}'-F_t'G_tm_{t-1})'\Sigma_{t}^{-1}\right\}
\\ && \times |\Sigma_{t}|^{-(p^{-1}\textrm{tr}(\beta) n+2p)/2}
\textrm{etr}\left(-\frac{1}{2} \beta^{1/2}
S_{t-1}\beta^{1/2}\Sigma_{t}^{-1}\right)
\\ &=& |\Sigma_{t}|^{-(p^{-1}\textrm{tr}(\beta) n+1+2p)/2} \textrm{etr} \bigg[
-\frac{1}{2} \big\{
(y_{t}'-F_t'G_tm_{t-1})Q_{t}^{-1}(y_{t}-m_{t-1}'G_t'F_t) \nonumber
\\ &&+\beta^{1/2} S_{t-1}\beta^{1/2}\big\}\Sigma_{t}^{-1}\bigg],
\end{eqnarray*}
which is proportional to the inverted Wishart distribution
$\IW_p(n^*+2p,S_t)$, with
\begin{gather}
S_t=\beta^{1/2} S_{t-1}\beta^{1/2}+e_{t}Q_{t}^{-1}e_{t}',\quad
n^*=p^{-1}\textrm{tr}(\beta) n+1,\label{eq30}
\end{gather}
where $e_{t}=y_{t}-y_{t-1}(1)=y_{t}-m_{t-1}'G_t'F_t$ is the one-step
forecast error and $Q_{t}=F_t'R_{t}F_t+1$. The recursions of $m_{t}$
and $P_{t}$ are calculated routinely, by writing down the posterior
distribution of $\Theta_{t}| \Sigma_{t},y^{t}$, i.e. $\Theta_{t}|
\Sigma_{t},y^{t}\sim \N_{p\times 1}(m_{t},P_{t},\Sigma_{t})$, where
from an application of the Kalman filter, we have
$m_{t}=G_tm_{t-1}+R_{t}F_tQ_{t}^{-1}e_{t}'$ and
$P_{t}=R_{t}-R_{t}F_tQ_{t}^{-1}F_t'R_{t}$.

The second parameter of the singular multivariate beta distribution
(see Lemma \ref{lemma} in Appendix), denoted by $q$, needs to
satisfy two requirements (a) $2q$ must be positive integer number
and (b) $p^{-1}\textrm{tr}(\beta)n+p$ must equal $n+p-1$. (a) is
needed for the singular multivariate beta distribution to be defined
(Uhlig, 1994) and (b) is needed for the distribution of the prior
Wishart of $\Phi_t|y^{t-1}$ (see Proposition \ref{pro1} in the
appendix). These two requirements result to the adoption of the
prior
$$
n=\frac{1}{1-p^{-1}\textrm{tr}(\beta)},
$$
where $\beta$ may be close, but not equal to $I_p$. With the above
prior of $n$, the degrees of freedom of equation (\ref{eq30}) become
$$
n^*=p^{-1}\textrm{tr}(\beta)n+1 =
\frac{1}{1-p^{-1}\textrm{tr}(\beta)}=n.
$$

Define $r_t=y_t-m_t'F_t$, the residual error vector. Then we have
that
\begin{displaymath}
r_t=e_t\{I_q-(R_tF_tQ_t^{-1})'F_t\}=e_tQ_t^{-1}(Q_t-F_t'R_tF_t)=
e_tQ_t^{-1}.
\end{displaymath}
From this, it follows that equation (\ref{eq30}) can be written as
$S_t=\beta^{1/2} S_{t-1}\beta^{1/2}+r_{t}e_{t}'$, or
\begin{equation}\label{eq32}
S_t= \beta^{t/2}S_0\beta^{t/2} + \sum_{i=0}^{t-1}\beta^{i/2}
r_{t-i}e_{t-i}'\beta^{i/2}.
\end{equation}
The posterior expectation of $\Sigma_t$ is
$\E(\Sigma_t|y^t)=S_t/(n-2)=(1-p^{-1}\textrm{tr}(\beta))S_t/(2p^{-1}\textrm{tr}(\beta)-1)$,
for $p^{-1}\textrm{tr}(\beta)>1/2$. From equation (\ref{prior:v}),
the one-step forecast mean of $\Sigma_t$ is
$\E(\Sigma_t|y^{t-1})=(1-p^{-1}\textrm{tr}(\beta))\beta^{1/2}
S_{t-1}\beta^{1/2}/(3p^{-1}\textrm{tr}(\beta)-2)$, for
$p^{-1}\textrm{tr}(\beta)>2/3$.

The above estimation procedure is valid for $0<\beta_i< 1$, while
from equation (\ref{vol:evol}) if $\beta_1=\beta_2=\ldots
=\beta_p=1$, then $\Phi_t=\Phi_{t-1}$ and the volatility is
unchanged from $t-1$ to $t$. Note that if
$\beta_1=\beta_2=\ldots=\beta_p$ we have $\beta=\beta_1I_p$ and in
this special case all elements of $\Sigma_t$ are discounted in the
same rate. The advantage of employing the discount matrix $\beta$ is
that different elements of the volatility estimator $\Sigma_t$ can
be discounted at different rate. For example for $p=2$, one can set
$\beta=\textrm{diag}(1,0.9)$, so that with
$\Sigma_t=(\sigma_{ij,t})_{i,j=1,2}$, the variance $\sigma_{11,t}$
has constant volatility, but the variance $\sigma_{22,t}$ is
discounted at a rate according to a discount factor of 0.9. The
situation $\beta=I_p$, is leading to a time-invariant volatility
$\Sigma_t=\Sigma$, for all $t$, and this is usually impractical. In
this case, the posterior distribution of $\Sigma$ is the inverted
Wishart $\Sigma|y^t\sim \IW_p(n_0+t+2p,S_t)$, with
$S_t=S_0+\sum_{i=1}^tr_ie_i'$, where $n_0$ are the initial degrees
of freedom. In the next result we relate the above posterior
estimate $S_t$ with the maximum likelihood estimator of $\Sigma$.
First note that conditional on $\Sigma$, the posterior distribution
of $\Theta_t$ is
\begin{equation}\label{dwr5a}
\Theta_t|\Sigma,y^t\sim \N_{d\times p}(m_t,P_t,\Sigma).
\end{equation}
Then we have the following result.
\begin{thm}\label{th0b}
In the MV-DLM (\ref{model2}) suppose that, for all $t$,
$\Sigma_t=\Sigma$, and so conditional on $\Sigma$, the posterior
distribution of $\Theta_t$ is given by equation (\ref{dwr5a}). Then
the maximum likelihood estimator of $\Sigma$, based on data
$y^N=\{y_1,y_2,\ldots,y_N\}$, is
$$
\widehat{\Sigma}_N=\frac{1}{N}\sum_{t=1}^{N}r_te_t',
$$
where $e_t=y_t-m_{t-1}'G_t'F_t$ is the one-step forecast error
vector and $r_t=y_t-m_t'F_t$ is the residual error vector.
\end{thm}
For $n_0=0$ and $S_0=0$, the estimator of $\Sigma$, which results
from the above inverted Wishart prior is
$S_N=N^{-1}\sum_{t=1}^Nr_te_t'=\widehat{\Sigma}_N$ and so the
posterior estimator of $\Sigma$ equals to the maximum likelihood
estimator of $\Sigma$. However, when $\Sigma_t$ is a time-dependent
volatility matrix, a similar procedure for the maximum likelihood
estimator of $\Sigma_t$ is not available in closed form and so the
above sequential Bayesian estimation procedure is thought to be
advantageous and preferable as opposed to approximate likelihood
estimation procedures (Durbin and Koopman, 2001). The log-likelihood
function when $\Sigma_t$ is time-dependent is given in Theorem
\ref{th2} of the next section.

\section{Model Diagnostics and Model Comparison}\label{test}

From equation (\ref{prior:v}) we have that the one-step forecast
mean of $\Sigma_t$ is
$\E(\Sigma_t|y^{t-1})=(1-p^{-1}\textrm{tr}(\beta))\beta^{1/2}
S_{t-1} \beta^{1/2} / (3p^{-1}\textrm{tr}(\beta)-2)$, where
$p^{-1}\textrm{tr}(\beta)>2/3$. The one-step forecast error
distribution is a $p$-variate $t$ distribution, i.e.
$$
e_t|y^{t-1}\sim \T_{p\times 1}(k, 0,
Q_t\beta^{1/2}S_{t-1}\beta^{1/2})
$$
where $k=p^{-1}\textrm{tr}(\beta)/(1-p^{-1}\textrm{tr}(\beta))$.
Note that the condition $p^{-1}\textrm{tr}(\beta)>2/3$ ensures that
$k>2$, hence, given $y^{t-1}$, the covariance matrix of $e_t$
exists. By defining
$$
u_t=(Q_t^*)^{1/2}e_t=\{(k-2)Q_t^{-1}\beta^{-1/2}S_{t-1}^{-1}\beta^{-1/2}\}^{1/2}e_t,
$$
the one-step standardized forecast errors, we obtain
$$
u_t|y^{t-1}\sim \T_{p\times 1}\{k, 0, (k-2)I_p\},
$$
where $(Q_t^*)^{1/2}$ denotes the square root of $Q_t^*$, based on
the Choleski decomposition, or based on the spectral decomposition.
From this it follows that $\E(u_t|y^{t-1})=0$ and
$\var(u_t|y^{t-1})=\E(u_tu_t'|y^{t-1})=I_p$ and so, by writing
$u_t=[u_{1t}~u_{2t}~\cdots~u_{pt}]'$, one measure of goodness of fit
is the mean of squared standardized one-step forecast errors (MSSE),
defined by
$$
\textrm{MSSE}=\frac{1}{N}\sum_{t=1}^N \left[ u_{1t}^2 ~ u_{2t}^2 ~
\cdots ~ u_{pt}^2 \right]',
$$
which should be close to $[1~1~\cdots~1]'$, if the model produces a
good fit to the data. Of course when $\beta=I_p$, the above $t$
distributions can not be defined, since $\textrm{tr}(\beta)=p$. In
this case we have $e_t|y^{t-1}\sim\T_{p\times
1}(n_0+t-1,0,Q_tS_{t-1})$ and then, with $k_t=n_0+t-1$, we get
$u_t|y^{t-1}\sim\T_{p\times 1}\{k_t, 0, (k_t-2)I_p\}$ and hence all
other definitions remain unchanged. Other measures of goodness of
fit are the mean absolute one-step forecast errors (MAE) and mean
error (ME), defined, respectively, by
$$
\textrm{MAE}=\frac{1}{N}\sum_{t=1}^N\left[\textrm{mod}(e_{1t}) ~
\textrm{mod}(e_{2t}) ~ \cdots ~ \textrm{mod}(e_{pt})\right]'\quad
\textrm{and} \quad \textrm{ME}=\frac{1}{N}\sum_{t=1}^Ne_t,
$$
where $e_t=[e_{1t}~e_{2t}~\cdots~e_{pt}]'$ and
$\textrm{mod}(e_{it})$ denotes the modulus of $e_{it}$, for
$i=1,2,\ldots,p$.

Another method of model diagnostics and model comparison is based on
the Value-at-Risk (VaR), which in laid words is the amount of money
of an asset that one expects to lose with some probability over a
certain time horizon. There are several ways of calculating the VaR
of a portfolio, but here we mention only the most popular, which is
termed as the variance-covariance approach and it is due to Morgan
(1996). The VaR of a portfolio has a single value (under a specific
model), which according to Brooks and Persand (2003) is
$$
\textrm{VaR}(N,\alpha)=\mu_N+F_N^{-1}(1-\alpha/100)\sigma_N,
$$
where $\textrm{VaR}(N,\alpha)$ is the VaR of a portfolio at time $N$
and percentage significance level $\alpha$, $F_N(\cdot)$ is the
distribution function of the standardized portfolio returns
$(z_N-\mu_N)/\sigma_N$, and $\sigma_N^2$ is the conditional
volatility of $z_N$. For known weights $w_1,\ldots,w_p$ satisfying
$w_i\geq 0$ and $\sum_{i=1}^pw_i=1$, we define the portfolio returns
$z_t=\sum_{i=1}^pw_ix_{i,t}$ and so its volatility is
$\sigma_t^2=\sum_{i=1}^pw_i^2\sigma_{ii,t}+2\sum_{i<j}w_iw_j\sigma_{ij,t}$,
where $\Sigma_t=(\sigma_{ij})_{i,j=1,2,\ldots,p}$. For their
internal evaluation of market risk, investment banks typically use
$95\%$ significance levels, leading to less tight evaluation of VaR,
i.e. the resulting from VaR amount of money will cover $95\%$ of
probable loses. The Basel Committee on Banking Supervision (1996,
1998) uses a tight $99\%$ confidence percentage to ensure coverage
of $99\%$ losses. Clearly
$\textrm{VaR}(N,0.95)<\textrm{VaR}(N,0.99)$, since there is needed
more money to cover larger proportion of probable loses. More
details on VaR and its evaluation may be found in Tsay (2002,
Chapter 7) and Chong (2004).

Another measure of goodness of fit, is based on the evaluation of
the log-likelihood function, as a means of model design (e.g.
choosing values of the discount matrices $\Delta$ and $\beta$) and
model comparison. The next result gives an expression of the
log-likelihood function.

\begin{thm}\label{th2}
In the MV-DLM (\ref{model2}) denote with
$\ell(\Sigma_1,\Sigma_2,\ldots,\Sigma_N;y^N)$ the log-likelihood
function of $\Sigma_1,\Sigma_2,\ldots,\Sigma_N$, based on data
$y^N=\{y_1,y_2,\ldots,y_N\}$. Then it is
\begin{eqnarray*}
\ell(\Sigma_1,\Sigma_2,\ldots,\Sigma_N;y^N)  &=& c -
\frac{1}{2}\sum_{t=1}^N \bigg[ p\log Q_t + (p-m) \log |\Sigma_{t-1}|
+ e_t'Q_t^{-1}\Sigma_t^{-1}e_t \\ && + p\log |L_t| + (m-p-2) \log
|\Sigma_t| \bigg]
\end{eqnarray*}
and
\begin{eqnarray*}
c&=&\frac{N(m-p)}{2}\sum_{i=1}^p\log\beta_i +
N\log\Gamma_p\{(m+1)/2\}-\frac{pN}{2}\log 2-pN\log\pi - N\log
\Gamma_p(m/2),
\end{eqnarray*}
where $\beta=\textrm{diag}(\beta_1,\beta_2,\ldots,\beta_p)$,
$m=p^{-1}\textrm{tr}(\beta)/(1-p^{-1}\textrm{tr}(\beta))+p-1$ and
$L_t$ is the diagonal matrix with diagonal elements the positive
eigenvalues of
$I_p-(C_{t-1}')^{-1}\beta^{1/2}\Sigma_t^{-1}\beta^{1/2}C_{t-1}^{-1}$,
with $\Sigma_t^{-1}=C_t'C_t$.
\end{thm}
Note that if $\beta=I_p$, then $\Sigma_t=\Sigma$, for all $t$, and
the log-likelihood function of $\Sigma$ reduces to
\begin{equation}\label{loglsmall}
\ell(\Sigma;y^N)=-\frac{pN}{2}\log
(2\pi)-\frac{p}{2}\sum_{t=1}^N\log Q_t -\frac{N}{2}\log
|\Sigma|-\frac{1}{2}\sum_{t=1}^Ne_t'Q_t^{-1}\Sigma^{-1}e_t.
\end{equation}
The log-likelihood function of Theorem \ref{th2} is clearly provided
conditional on the values of $\Delta$ and $\beta$ and so, replacing
$\Sigma_t$ by $S_t/(n-2)$ (the posterior mean of $\Sigma_t$) in the
log-likelihood, one way to choose these values is by maximizing the
log-likelihood over a range of candidate values for $\Delta$ and
$\beta$.

In model comparison, the log-likelihood function is particularly
useful, as it can be used forming likelihood ratios in order to
compare and contrast the performance of two models. A similar idea
can be implemented by considering sequential model monitoring, for
which, two models are compared by using sequential Bayes' factors of
the standardized errors $u_1,u_2,\ldots,u_N$. Following the ideas of
West and Harrison (1997, Chapter 11) and Salvador and Gargallo
(2004), we consider two models $\mathcal{M}_1$ and $\mathcal{M}_2$,
which differ in some quantitative form, e.g. in the values of the
discount matrices, and by writing all densities conditional on these
two models, we form the log Bayes' factor
$$
LBF(t)=\log\left[\frac{p(u_t|y^{t-1},\mathcal{M}_1)}{p(u_t|y^{t-1},\mathcal{M}_2)}\right],
\quad t=1,2,\ldots,N.
$$
Then, at time $t$, $\mathcal{M}_1$ is in favour of $\mathcal{M}_2$
(equiv. $\mathcal{M}_2$ is in favour of $\mathcal{M}_1$), if
$LBF(t)>0$ (equiv. $LBF(t)<0$), while when $LBF(t)=0$, the two
models are equivalent, in the sense that they produce similar
forecasts and similar standardized forecast errors. Some algorithms
have been proposed in the literature about how the above test can be
done efficiently. Some work includes Monte Carlo simulation
(Salvador {\it et al.}, 2004), some work is restricted in the case
of a time-invariant volatility matrix (Salvador and Gargallo, 2004)
and most of the work refers to univariate processes (West and
Harrison, 1997, Salvador and Gargallo, 2004, 2005, 2006).
Triantafyllopoulos (2006b) proposes a general procedure, according
to which, a modified exponentially weighted moving average control
chart is applied to the univariate process
$\{LBF(t)\}_{t=1,2,\ldots,N}$ and control signals indicate model
preference.

The above ideas of model comparison, based on Bayes' factors, can
also be applied to the problem of sequential monitoring of a single
model. This approach, which is explored in detail in West and
Harrison (1997, Chapter 11) and in Salvador and Gargallo (2004,
2005, 2006), proposes the adoption of a set of alternative models,
compares the current model with these and makes a sequential
decision adopting the best model, according to the behaviour of the
Bayes' factor.

\section{Example: The London Metal Exchange Data}\label{data}

\subsection{Description of the Data}

The London metal exchange (LME) is the world's premier non-ferrous
metals market, with highly liquid contracts. Its trading customers
may be metal industries or individuals (sellers or buyers). The
metals currently traded in the exchange are: aluminium, copper grade
A, standard lead, primary nickel, tin, and zinc. More details about
the LME can be found on its web site: {\tt http://www.lme.co.uk}.

The importance of the LME and its operations has recently invited
considerable interest. Here, from a statistical point of view, we
mention the work of McKenzie {\it et al.} (2001) and the review of
Watkins and McAleer (2004). Triantafyllopoulos (2006a) gives a brief
account to the statistical work on the LME.

\begin{figure}[t]
 \epsfig{file=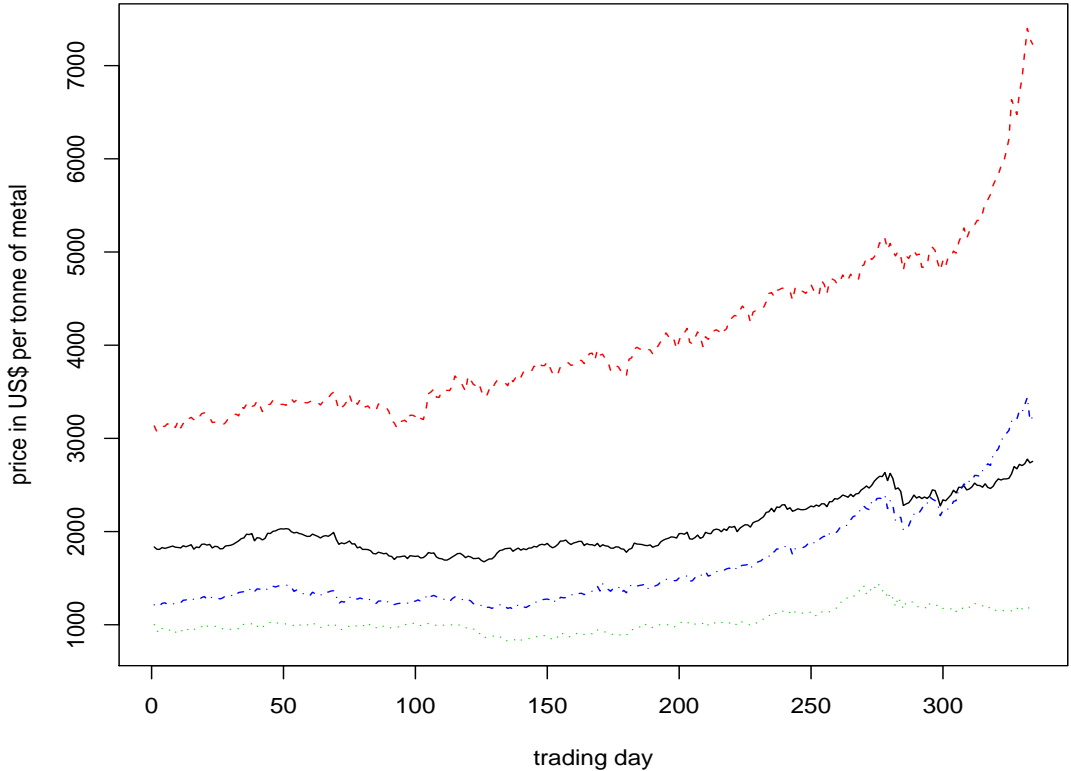, height=12cm, width=15cm}
 \caption{LME data, consisting of aluminium (solid line), copper (dashed line),
 lead (dotted line) and zinc (dashed-dotted line) spot prices (in US dollars
 per tonne of each metal).}\label{fig1}
\end{figure}

In this paper we concentrate on spot prices of four metals exchanged
in the LME, namely aluminium, copper, lead and zinc. We have 4
variables of interest collected in the observation vector
$y_t=[y_{1t}~y_{2t}~y_{3t}~y_{4t}]'$. Each variable comprises the
spot price per tonne of metal: $y_{1t}$ is the spot variable, which
indicates the daily/current ask price per tonne of aluminium; the
remaining three variables are the relevant spot ask prices of
copper, lead and zinc, respectively. The data are collected for
every trading day from 4 January 2005 to 28 April 2006, and are
plotted in Figure \ref{fig1}. After excluding week-ends and bank
holidays, there are $N=334$ trading days. The data have been
obtained from the LME web site: {\tt http://www.lme.co.uk}.

\subsection{Statistical Analysis}\label{analysis}

Here we consider the compound return time series
$\{x_t\}_{t=1,2,\ldots,333}$ with $x_{t-1}=\log y_t-\log y_{t-1}$,
for $t=2,3,\ldots,334$. Most of the current literature in
econometrics is focused on modelling only the volatility of the
series, but for the MV-DLMs considered in this paper, one can model
with the same model the returns (for forecasting purposes) and
estimate the volatility matrix.

We use the model
\begin{equation}\label{example}
x_t=\mu_t+\epsilon_t, \quad \mu_t=\Theta_t'F, \quad
\Theta_t=\Theta_{t-1}+\omega_t,
\end{equation}
where $\epsilon_t\sim N_{4\times 1}(0,\Sigma_t)$, $\omega_t\sim
N_{2\times 4}(0, \Omega_t, \Sigma_t)$ and $\mu_t$ is the level of
the series at time $t$. The design vector $F=[1~0]'$ is invariant of
time and a random walk evolution for the states $\Theta_t$ has been
chosen, which is suitable for modelling the compound returns (Tsay,
2002, Cuaresma and Hlouskova, 2005). The volatility of the series is
measured with the volatility matrix $\Sigma_t$, which is subject to
estimation. There might be some uncertainty on the dimension $d$ of
the rows of $\Theta_t$, but here for parsimonious modelling we
choose a low value for $d$. It might be worthwhile to consider $d$
as random, but this can add computational delays to the estimation
process. The $2\times 2$ evolution covariance matrix $\Omega_t$ can
be specified with two discount factors $\delta_1$ and $\delta_2$
according to the discussion in Section \ref{model:dwr}. However, it
can be seen that since $F$ is time invariant model (\ref{example})
can be decomposed as
$$
x_t=\mu_t+\epsilon_t,\quad \mu_t=\mu_{t-1}+\zeta_t, \quad
\zeta_t\sim N_{p\times 1}(0,F'\Omega_tF\Sigma_t),
$$
which is a random walk plus noise model. Since
$F'\Omega_tF=(1-\delta_1)p_{11,t-1}/\delta_1$, where
$P_t=(p_{ij,t})_{i,j=1,2,3,4}$, it can be seen that only $\delta_1$
has a contribution to the model and in particular model
(\ref{example}) is equivalent to a model with a single discount
factor, i.e. $\Omega_t=(1-\delta)P_{t-1}/\delta$ and
$\delta=\delta_1$. So there are five discount factors of interest:
$\delta$, which is the discount factor responsible for the random
walk evolution of the level $\mu_t$, and
$\beta_1,\beta_2,\beta_3,\beta_4$, which are responsible for the
evolution of the $4\times 4$ volatility matrix $\Sigma_t$; $\beta_i$
is the discount factor for the volatility of the compound series
$x_i$, where $\beta=[\beta_1~\beta_2~\beta_3~\beta_4]'$ and
$x_t=[x_{1t}~x_{2t}~x_{3t}~x_{4t}]'$. We specify the priors
$$
\Theta_0|\Sigma_0\sim \N_{2\times 4}(0,1000I_2,\Sigma_0),\quad
\mu_0|\Sigma_0\sim \N_{4\times 1}(0,1000\Sigma_0), \quad
\Sigma_0\sim \IW_4(n+8, I_4),
$$
where $n=1/(1-4^{-1}\textrm{tr}(\beta))$.

Table \ref{table1} shows two performance measures, namely the MSSE
and the log-likelihood function (see Section \ref{test}). Two values
of $\delta$ are picked and compared with; a small value
$\delta=0.08$ (corresponding to an adapting, but not smooth
evolution for the level $\mu_t$) and a high value $\delta=0.8$
(corresponding to a smooth evolution for the level $\mu_t$). The ME
was found to be constant throughout the range of $\beta$, but
changing for the two values of $\delta$; for $\delta=0.08$ it was
$\textrm{ME}=[0.04~-0.18~0.00~0.05]'$ and for $\delta=0.8$ it was
$\textrm{ME}=[0.21~1.22~0.17~0.01]'$. From Figure \ref{fig1} it is
apparent that the aluminium and the zinc evolve together (their
difference appears to be in their levels) and likewise the copper
and the lead evolve together. This can be reflected in our model by
choosing $\beta_1=\beta_4$ and $\beta_2=\beta_3$ so that the
volatilities of say aluminium and zinc will be similar. Table
\ref{table1} shows the two performance measures (MSSE and LogL) for
a range of admissible values of $\beta_1$ and $\beta_2$, given that
$\textrm{tr}(\beta)/4>2/3$ so that the one-step forecast mean of
$\Sigma_t$ exists (see Section \ref{estimation}). For all $\beta\neq
I_4$ and for $\delta=0.08$, the log-likelihood function is maximized
for $\beta_1=\beta_2=\beta_3=\beta_4=0.2$
$(\textrm{LogL}=-11421.3)$, but this value can not be allowed,
because $(0.2+0.2+0.2+0.2)/4=0.2<2/3$. The highest value of LogL is
achieved for $\beta_1=0.65$ and $\beta_2=0.7$, but this produces
poor performance in the MSSE. Our choice is for $\beta_1=0.66$ and
$\beta_2=0.9$, returning reasonable values of the MSSE and a not
very low value for the LogL. When comparing the performance of the
models for the discount factors $\delta=0.08$ and $\delta=0.8$, we
note that the log-likelihood function corresponding to $\delta=0.8$
is smaller than that of $\delta=0.08$. Similarly the ME produced
with $\delta=0.8$ is too large and the MSSE does not achieve a
decent value for all four variables. Therefore we conclude that a
high discount factor $\delta$ should not be chosen.

Table \ref{table1} also reveals that a choice of
$\beta_1=\beta_2=\beta_3=\beta_4$ is inadequate, leading to poor
performance in the MSSE. It is clear that there are two main factors
driving the volatilities of the metals and these factors are
expressed here by the two discount factors $\beta_1$ and $\beta_2$.
The log-likelihood function for $\beta=[1~1~1~1]'$ (e.g. when
$\Sigma_t=\Sigma$), is $-\infty$ when the formula of Theorem
\ref{th2} is used (due to the infinity at the value of $m$), but
this likelihood is just -10344.66 when formula (\ref{loglsmall}) is
used. The fact that this log-likelihood appears to be the maximum
likelihood, is due to the fact that in the likelihood
(\ref{loglsmall}) the part of $\log p(\Sigma_t|\Sigma_{t-1})$ does
not appear. Likelihood (\ref{loglsmall}) should only be used when
there is strong evidence to suggest that the volatility is constant,
which clearly is not the case in this data set.

\begin{table}
\caption{Mean square one-step forecast standardized errors (MSSE)
and log-likelihood function (LogL) evaluated at the posterior mean
$S_t/(n-2)$, where
$\beta=[\beta_1~\beta_2~\beta_2~\beta_1]'$.}\label{table1}
\begin{center}
\begin{tabular}{|c||cccc||cccc||cc|}
\hline & MSSE & & & & MSSE & & &  & LogL & LogL \\ \hline &
$\delta=0.08$ & & & & $\delta=0.8$ & & & & $\delta=0.08$ &
$\delta=0.8$ \\ $[\beta_1~\beta_2]'$ &
Alum & Copp & Lead & Zinc &  Alum & Copp & Lead & Zinc &  &  \\
\hline

$0.65~0.70$ & 1.11 & 2.82 & 2.71 & 0.98 & 2.83 & 14.07 & 14.24 &
1.87 & -23980.40 & -29493.89  \\
$0.70~0.70$ & 0.86 & 2.83 & 2.70 & 0.73 & 2.33
& 14.11 & 14.20 & 1.46 & -25398.02 & -31200.99 \\
$0.75~0.70$ & 0.66 & 2.84 & 2.70 & 0.54 & 1.94 & 14.14 &
14.17 & 1.15 & -27088.82 & -33233.69 \\
$0.80~0.70$ & 0.51 & 2.84 & 2.69 & 0.40 & 1.62 & 14.17 & 14.14 &
0.92 & -29155.95 & -35711.21 \\
$0.85~0.70$ & 0.38 & 2.85 & 2.69 & 0.29 & 1.37 & 14.19 & 14.11 &
0.73 & -31766.03 & -38824.03 \\
$0.90~0.70$ & 0.29 & 2.86 &
2.69 & 0.20 & 1.16 & 14.21 & 14.10 & 0.58 & -35220.89 & -42910.13 \\
$0.95~0.70$ & 0.21 &
2.87 & 2.69 & 0.13 & 1.00 & 14.23 & 14.08 & 0.46 & -40202.95 & -48706.83 \\
$1.00~0.70$ & 0.16 & 2.87 & 2.69 & 0.08 & 0.86 & 14.25 & 14.07 &
0.37 & -49796.86 & -59413.16 \\
$0.60~0.80$ & 1.44 & 1.93 & 1.87 & 1.32 & 3.49 & 10.22 & 10.47 &
2.44 & -25666.85 & -31455.59 \\ $0.70~0.80$ & 0.86 & 1.95 & 1.86 &
0.74 & 2.33 & 10.29 & 10.41 & 1.48 & -29182.35 & -35726.44 \\
$0.80~0.80$ & 0.51 &
1.96 & 1.89 & 0.41 & 1.62 & 10.34 & 10.36 & 0.93 & -34531.82 & -42208.33 \\
$0.90~0.80$ & 0.29 & 1.97 & 1.85 & 0.20 & 1.16 & 10.37 & 10.32 &
0.59 & -43944.85 & -53515.26 \\ $1.00~0.80$ & 0.16 & 1.98 & 1.86 &
0.08 & 0.86 & 10.40 & 10.30 & 0.38 & -68721.06 & -82127.66 \\
$0.50~0.90$ & 2.42 &
1.35 & 1.26 & 2.41 & 5.50 & 7.66 & 8.06 & 4.27 & -26661.76 & -32427.18 \\
$0.60~0.90$ & 1.42 & 1.37 & 1.35 & 1.32 & 3.48 & 7.73 & 7.98 & 2.45
& -30145.39 & -36667.82 \\ $0.66~0.90$ & 1.05 & 1.37 & 1.34 & 0.94 &
2.72 &
7.76 & 7.95 & 1.81 & -32970.79 & -40111.85 \\

$0.70~0.90$ & 0.86 & 1.38 & 1.34 & 0.75 & 2.34 & 7.78 & 7.93 & 1.49
& -35324.88 & -42981.32 \\ $0.80~0.90$ & 0.52 &
1.39 & 1.34 & 0.42 & 1.63 & 7.82 & 7.89 & 0.94 & -44033.40 & -53579.88 \\
$0.90~0.90$ & 0.29 & 1.40 & 1.33 & 0.21 & 1.16 & 7.85 & 7.86 & 0.60
& -62082.88 & -75419.47 \\ $1.00~0.90$ & 0.16 & 1.41 & 1.33 & 0.08 &
0.85 & 7.87 & 7.84 & 0.38 & -126623.4 & -151414.0 \\ $0.40~1.00$ &
4.37 &
0.97 & 1.04 & 4.75 & 9.38 & 5.91 & 6.42 & 8.28 & -31289.33 & -36988.89 \\
$0.50~1.00$ & 2.39 & 0.98 & 1.02 & 2.42 & 5.45 & 5.98 & 6.34 & 4.29
& -35204.39 & -41657.98 \\ $0.60~1.00$ & 1.42 & 1.00 & 1.01 & 1.34 &
3.46 & 6.03 & 6.28 & 2.46 & -40956.66 & -48538.06 \\ $0.66~1.00$ &
1.05 &
1.00 & 1.00 & 0.95 & 2.72 & 6.06 & 6.25 & 1.82 & -45996.54 & -54572.35 \\

$0.70~1.00$ & 0.87 & 1.01 & 1.00 & 0.76 & 2.33 & 6.07 & 6.24 & 1.51
& -50472.62 & -59932.27 \\ $0.80~1.00$ & 0.52 & 1.01 & 0.99 & 0.43 &
1.63 & 6.10 & 6.20 & 0.95 & -69578.42 & -82795.49 \\ $0.90~1.00$ &
0.30 &
1.02 & 0.99 & 0.22 & 1.16 & 6.13 & 6.18 & 0.61 & -127928.4 & -152439.1 \\
$1.00~1.00$ & 0.16 & 1.02 & 0.98 & 0.08 & 0.85 & 6.14 & 6.16 & 0.38
& -10344.66 & -9947.38 \\
\hline
\end{tabular}
\end{center}
\end{table}

Table \ref{table2} shows the evaluation of VaR based on the
variance-covariance approach (see Section \ref{test}) for several
values of $\beta_1$ and $\beta_2$ and for $\delta=0.08$,
$\delta=0.8$ and $\delta=1$. Typically a $95\%$ confidence level is
used by investment banks and a $99\%$ confidence level is used by
the Basle Committee (Chong, 2004). $\delta=1$ refers to a
time-invariant level $\mu_t=\mu$, which is adopted in many MGARCH
type models (Bauwens {\it et al.}, 2006), while $\delta=0.8$
generates a time-dependent, but smooth level, and $\delta=0.08$
generates a highly adaptive time-dependent level $\mu_t$. Table
\ref{table2} shows that, for the same parameters of $\beta$, the VaR
using $\delta=1$ and $\delta=0.8$ are larger as compared with
$\delta=0.08$. Within $\delta=0.08$, the parameters
$\beta_1=0.7,\beta_2=0.8$, $\beta_1=0.66,\beta_2=0.9$,
$0.9,\beta_2=1$ and $\beta_1=\beta_2=1$ result to the best models.
From Tables \ref{table1} and \ref{table2} we suggest that the
overall best model is this with $\beta_1=0.66,\beta_2=0.9$,
producing not very low log-likelihood function, a decent MSSE and a
relatively low values of VaR.

\begin{table}
\caption{$95\%$ and $99\%$ VaR values of the portfolio of
$x_N=[x_{1,N}~x_{2,N}~x_{3,N}~x_{4,N}]'$, for several values of
$\beta=\textrm{diag}(\beta_1,\beta_2,\beta_2,\beta_1)$, for
$\delta=0.08$, $\delta=0.8$ and $\delta=1$.}\label{table2}
\begin{center}
\begin{tabular}{|c||cccccc|}
\hline & $95\%$ & $99\%$ & $95\%$ & $99\%$ & $95\%$ & $99\%$  \\
\hline
$[\beta_1~\beta_2]'$ & $\delta=0.08$ & & $\delta=0.8$ & & $\delta=1$ & \\
\hline $0.65~0.70$ & 293.574 & 415.207 & 756.831 & 1070.401 &
745.768 & 1054.754 \\ $0.60~0.80$ & 151.528 & 214.309 & 387.083 &
547.459 & 397.888 & 562.741 \\ $0.70~0.80$ &
92.249 & 130.470 & 234.716 & 331.964 & 243.412 & 344.263 \\
$0.50~0.90$ & 167.536 & 236.950 & 422.821 & 598.004 & 456.304 &
645.360 \\ $0.66~0.90$ & 83.463 & 118.043 & 209.949 & 296.935 &
227.667 & 321.994 \\ $0.50~1.00$ & 242.451 & 342.903 & 576.401 &
815.215 & 613.341 & 867.460 \\ $0.66~1.00$ & 144.758 & 204.734 &
344.389 & 487.076 & 367.356 & 519.559 \\ $0.90~1.00$ & 62.820 &
88.847 & 149.211 & 211.032 & 161.045 & 227.770 \\ $1.00~1.00$ &
$21.361$
& $30.273$ & $49.863$ & $70.667$ & $53.291$ & $75.523$ \\
\hline
\end{tabular}
\end{center}
\end{table}

Figure \ref{fig3} shows the one-step forecast of the volatilities
(diagonal elements of $\Sigma_t$) and Figure \ref{fig4} shows the
respective forecasts of the correlations of $\Sigma_t$. Figure
\ref{fig3} illustrates that the volatilities of aluminium and zinc
have a similar pattern and the volatilities of copper and lead have
a similar pattern. Copper and zinc appear to be the most volatile
and this is expected if we look at Figure \ref{fig1}, where the
trends of copper and zinc are less smooth than those of aluminium
and zinc. Figure \ref{fig4} confirms that the aluminium and the zinc
are more correlated than the aluminium and the lead. This figure
also indicates that the correlations are not very high in modulus.

\begin{figure}[t]
 \epsfig{file=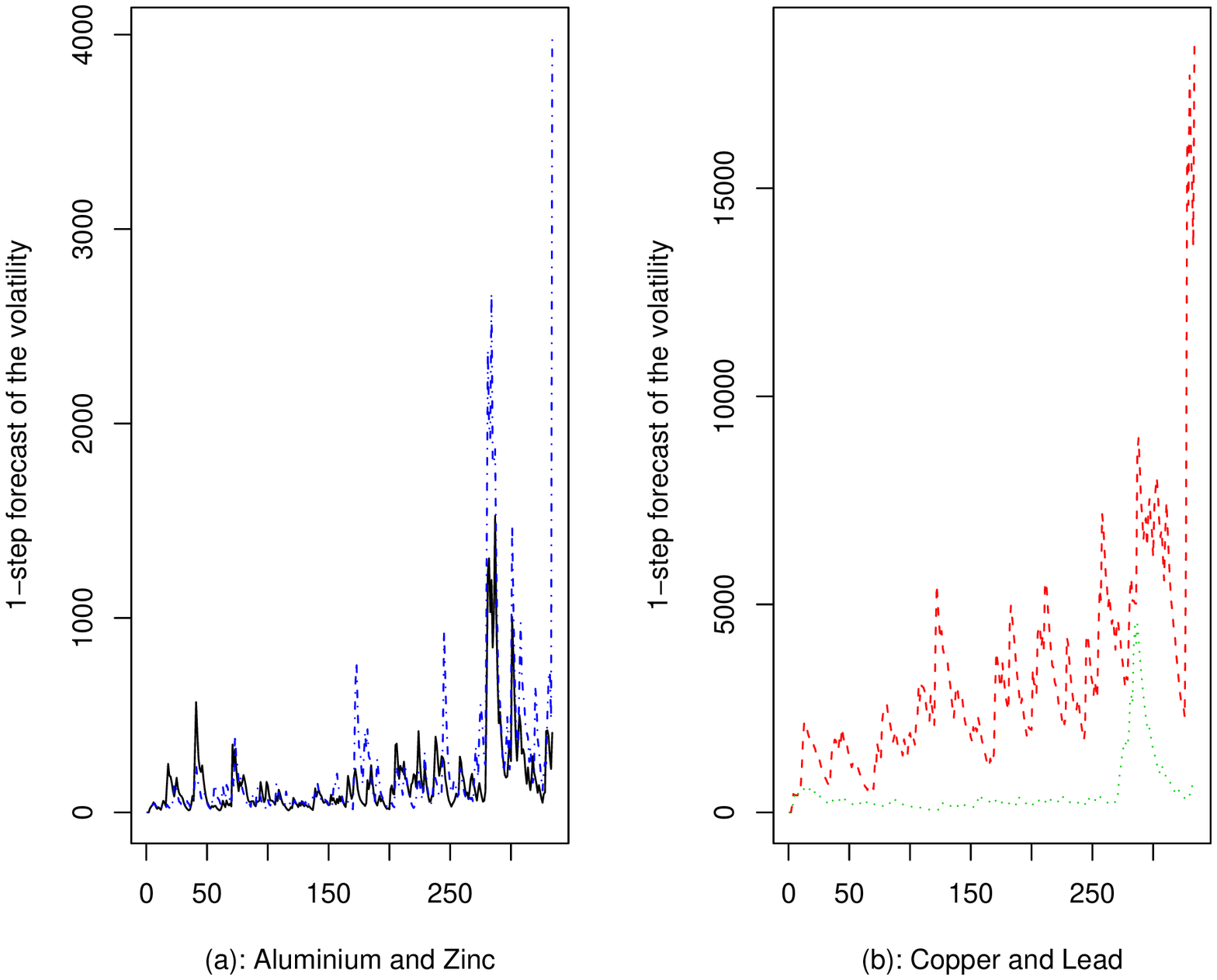, height=10cm, width=15cm}
 \caption{One-step forecasts of the volatility of aluminium (solid
 line of panel (a)), zinc (dashed-dotted line of panel (a)), copper
 (dashed line of panel (b)) and lead (dotted line of panel (b)).}\label{fig3}
\end{figure}

\begin{figure}
 \epsfig{file=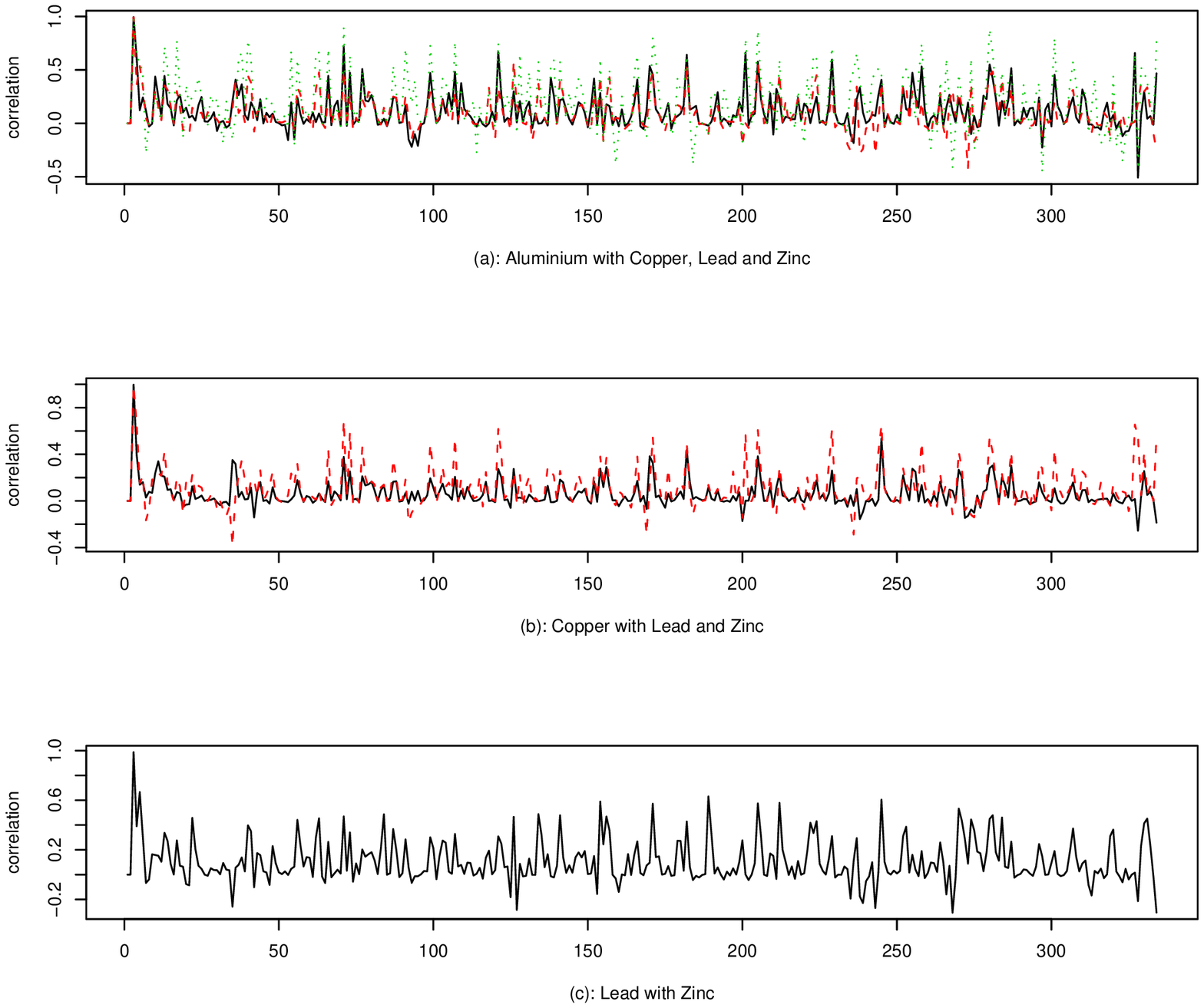, height=16cm, width=17cm}
 \caption{One-step forecast of the correlation of aluminium,
 copper, lead and zinc. Panel (a) shows the correlation of
 aluminium with copper (solid line), the correlation of aluminium
 with lead (dashed line) and the correlation of aluminium with zinc
 (dotted line); panel (b) shows the correlation of copper with lead
 (solid line) and the correlation of copper with zinc (dashed
 line); panel (c) shows the correlation of lead with zinc.}\label{fig4}
\end{figure}

From Figure \ref{fig1} we can clearly see that the aluminium and the
zinc are locally co-integrated of order 1, and the copper and lead
are also locally co-integrated of order 1. Here we use the term
\emph{locally co-integrated} of order $d$ to indicate that a linear
combination of each of the two variables are, after $d$ steps of
integration, locally stationary (in the sense that for a time
period, known also as \emph{regime} the time series is weakly
stationary). The aluminium and the copper are not co-integrated and
the same applies for the copper and zinc. This fact is apparent in
the volatilities (Figure \ref{fig3}) and in the model this is
reflected by the choice of two distinct elements in the discount
matrix $\beta$, i.e $\beta_1=\beta_4$ and $\beta_2=\beta_3$. There
are two distinct factors driving the volatilities of the four metals
and a factor volatility model could be applied to reduce the
complexity (Aguilar and West, 2000, Tsay, 2002, \S9.4).

\section{Discussion}\label{conclusions}

This paper develops a new Bayesian procedure for estimation and
forecasting of multivariate volatility. It is proposed that the
evolution of the unknown volatility covariance matrix is modelled
with a multiplicative stochastic model, based on Wishart and
singular multivariate beta distributions. The resulting algorithm is
capable to estimate the volatility element by element. This is
achieved by employing variance discounting using several discount
factors and thus allowing different volatilities to be discounted at
different rates.

In the last two decades many models have been developed for
multivariate volatility estimation (see Section \ref{intro}). Here
we provide a discussion of the advantages of our proposal compared
to the multivariate GARCH (MGARCH) models, reviewed in Bauwens {\it
et al.} (2006). Some of the MGARCH models result as generalizations
of univariate GARCH models (e.g. the VEC, the constant-correlation
GARCH, and the BEKK models, see also Section \ref{intro}). From
these models the constant-correlation GARCH model makes the strong
and usually unrealistic assumption of a constant correlation matrix,
whilst the VEC and even the BEKK have too many parameters to
estimate. The large number of parameters to be estimated, restrict
these models to applications of relatively low dimensions, usually
not exceeding $p=3$. The factor GARCH models (e.g the factor-BEKK
model) overcomes this difficulty, but in practice the specification
of the factors is not simple (Tsay, 2002, \S9.4). The
dynamic-correlation models (Bauwens {\it et al.}, 2006; Audrino and
Barone-Adesi, 2006) aim to combine the flexibility of the
constant-correlation GARCH, but to overcome the main drawback of
that model by introducing a specific time-dependent structure on the
correlation matrix. This can be done in several ways, but its main
drawback is that, if the dimension of the parameters is to be
manageable, the correlation matrix is driven by scalar parameters,
which means that all correlations have the same weight of change.
Perhaps, this is not a major issue for bivariate time series data,
but for higher dimensions it is unlikely to hold true. In our MV-DLM
model we overcome this problem by introducing the matrix of discount
factors $\beta$ and by discounting the volatilities and the
corresponding correlations at different rates.

The usual setting of a MGARCH model is that of
$y_t=\mu_t+\epsilon_t$, where $\mu_t$ is the level of the series
$y_t$ (usually the series will be the compound returns of some
assets or exchange rates), with $\epsilon_t$ being the innovation
series, following $\epsilon_t|\Sigma_t\sim\N_{p\times
1}(0,\Sigma_t)$ and $\Sigma_t$ represents the volatility matrix
subject to estimation. While it is recognized that the volatility
can affect the level, in some MGARCH studies the level is
time-invariant (Bauwens {\it et al.}, 2006), and in some other
studies the level is assumed to have a simple evolution, e.g. to
follow an autoregressive model of order one (Audrino and
Barone-Adesi, 2006). In the latter case estimation is usually
performed separately in the AR and GARCH components, which may not
be desirable for on-line forecasting. Our proposed model does in
fact allow for much more complicated structure in $\mu_t$, through
$\mu_t=\Theta_t'F_t$ and through the evolution equation of
$\Theta_t$, see equation (\ref{model2}). This can include structural
characteristics such as trend and seasonal components and applying
the principle of superposition of state space models (West and
Harrison, 1997, Chapter 6), one can build complex multivariate time
series models, for which estimation of the states is accompanied by
simultaneous estimation of the volatility. This is not achievable,
by neither ARIMA type models, nor by MGARCH models alone. In order
to build such models one has to consider a multivariate ARIMA model,
with errors following MGARCH models. In such models there are
inferential problems regarding to estimation and in the literature
simple models have been considered; for a univariate discussion on
this topic see Fiorentini and Maravall (1996) and Audrino and
Barone-Adesi (2006).

An important issue, which is discussed in Bauwens {\it et al.}
(2006), is that of marginalization. If $y_t$ follows a MGARCH model,
the question is whether $y_t^*=Ay_t$, follows the same type of
MGARCH model, where $A$ is a $q\times p$ matrix of constants. This
is an important problem, because if the model is closed under linear
transformations (or else if it is invariant under linear
transformations), then one can easily study the volatility of a
linear combination of some assets, for example to estimate the
volatility of a portfolio and hence the value at risk of a
portfolio. As pointed out in Bauwens {\it et al.} (2006) not all
GARCH models are invariant under the above linear transformation.
The MV-DLM is invariant, under some regulatory assumptions. Consider
model (\ref{model2}) and define $A$ a $q\times p$ matrix of rank
$q$. Then, we can write
$$
(y_t^*)'=F_t'\Theta_t^*+(\epsilon_t^*)',\quad
\Theta_t^*=G_t\Theta_{t-1}^*+\omega_t^*,\quad \epsilon_t^*\sim
\N_{q\times 1}(0, \Sigma_t^*), \quad \omega_t^*\sim\N_{d\times q}(0,
\Omega_t,\Sigma_t^*),
$$
where $\Theta_t^*=\Theta_tA'$, $\epsilon_t^*=A\epsilon_t$,
$\Sigma_t^*=A\Sigma_tA'$, $\omega_t^*=\omega_tA'$ and the remaining
components of the model is as in (\ref{model2}). Although, in the
above model we can not obtain an explicit formula for the precision
$\Psi_t=(A\Sigma_tA')^{-1}$, it is clear that using distribution
theory, we can establish that the linear transformation $y_t^*=Ay_t$
follows a MV-DLM with dimensions $q$ and $d$. For example, we can
readily see that from the posterior
$\Sigma_t|y^t\sim\IW_p(n+2p,S_t)$ we have
$\Sigma_t^*|(y^*)^t\sim\IW_q(n^*+2q,AS_tA')$, where $n^*=n+2(p-q)$.
It follows that all scalar $y_{it}$, with
$y_t=[y_{1t}~y_{2t}~\cdots~y_{pt}]'$, follow univariate DLMs of the
form of West and Harrison (1997, \S10.8) and the posterior
distributions of the diagonal elements
$\sigma_{11,t},\sigma_{22,t},\ldots,\sigma_{pp,t}$ of
$\Sigma_t=(\sigma_{ij,t})_{i,j=1,2,\ldots,p}$ are inverted gamma.

The Bayesian estimation approach of the MV-DLMs is preferred to the
usual maximum likelihood estimation approach of most of the MGARCH
models or to Bayesian estimation based on Monte Carlo simulation.
The proposed Bayesian approach is delivered in closed form and thus
it is available for on-line estimation. In the maximum likelihood
estimation approach, adopted in many MGARCH models, given a sample,
the aim is to estimate a set of parameters, sometimes a reasonably
large number of them and sometimes the maximization will be
computationally expensive and time consuming. This procedure may not
be suitable for sequential application, since the parameters and
their estimates seem to lose one of their dynamic power, which is to
adapt and to update as new information comes in. Our model is
adaptive to new information and it is computationally cheap, which
makes it suitable for volatility estimation of high dimensional
data.

\section*{Acknowledgements}
I am grateful to Giovanni Montana and to Tony O'Hagan, for several
helpful comments and suggestions on an earlier draft of the paper. I
wish to thank an anonymous referee, for providing detailed comments
that led to a considerably improved version of the paper.

\renewcommand{\theequation}{A-\arabic{equation}} 
\setcounter{equation}{0}  
\section*{Appendix}

In this appendix we detail the proofs of arguments in Sections
\ref{estimation} and \ref{test}. We begin with the prior
distribution (\ref{prior:v}).

\begin{pro}\label{pro1}
Consider model (\ref{model2}) with the priors (\ref{dwr3}) and the
evolution (\ref{vol:evol}). Let the posterior precision at time
$t-1$ be $\Phi_{t-1}|y^{t-1}\sim\W_p(n+p-1, S_{t-1}^{-1})$. Then,
with the prior degrees of freedom
$n=1/(1-p^{-1}\textrm{tr}(\beta))$, the prior distribution of
$\Phi_t$ is $\Phi_t|y^{t-1}\sim
\W_p(p^{-1}\textrm{tr}(\beta)n+p-1,\beta^{-1/2}S_{t-1}^{-1}\beta^{-1/2})$.
\end{pro}
The proof is a direct consequence of the model assumptions and
Theorem 1 of Uhlig (1994). From the above proposition, the prior
(\ref{prior:v}) is obtained from $\Sigma_t=\Phi_t^{-1}$.

\begin{proof}[Proof of Theorem \ref{th0b}]
The proof of the maximization of the log-likelihood function
requires matrix-differentiation, in particular, first and second
order differentiation in terms of $\Sigma$. Here we follow the
matrix-differentiation notation of Harville (1997) and the proof
mimics the early work on log-likelihood maximization of Harvey
(1986, 1989, \S8.3). An alternative proof can be obtained by
employing the log-likelihood maximization procedure, used for VAR
models, of L\"utkepohl (1993, pages 80-82).

With the posterior (\ref{dwr5a}), the forecast distribution of
$y_t|\Sigma$ is $y_t|\Sigma,y^{t-1}\sim \N_{p\times
1}(m_{t-1}'G_t'F_t,Q_t\Sigma)$, where $Q_t=F_t'R_tF_t+1$ and
$m_{t-1}$ and $R_t$ are defined in Section \ref{estimation}. The log
likelihood of $\Sigma$ is
\begin{eqnarray}
\ell (\Sigma ;y^N)&=&\log\prod_{t=1}^Np(y_t|\Sigma,y^{t-1}) =
-\frac{pN}{2}\log (2\pi)-\frac{p}{2}\sum_{t=1}^N\log Q_t
-\frac{N}{2}\log |\Sigma| \nonumber \\ && -
\frac{1}{2}\sum_{t=1}^N(y_t'-F_t'G_tm_{t-1})Q_t^{-1}\Sigma^{-1}(y_t-m_{t-1}'G_t'F_t).\label{app:loglb}
\end{eqnarray}
Taking the first derivative of $\ell(\Sigma;y^N)$ we get
\begin{eqnarray}
\frac{\partial \ell (\Sigma ;y^N)}{\partial \Sigma^{-1}} &=&
-\frac{N}{2} \frac{\partial\log |\Sigma^{-1}|^{-1}}{\partial
\Sigma^{-1}} - \frac{1}{2} \sum_{t=1}^N \frac{ \partial \{ (y_t' -
F_t'G_tm_{t-1})Q_t^{-1}\Sigma^{-1} (y_t-m_{t-1}G_t'F_t)\} } {
\partial \Sigma^{-1} } \nonumber \\ &=&
N\Sigma -
\frac{N}{2}\textrm{diag}\{\sigma_{11},\sigma_{22},\ldots,\sigma_{pp}\}
\nonumber \\ && - \frac{1}{2}\sum_{t=1}^N
\left(Q_t^{-1}e_te_t'-\textrm{diag}\left\{\frac{e_{1t}^2}{Q_t},\frac{e_{2t}^2}{Q_t},
\ldots,\frac{e_{pt}^2}{Q_t}\right\}\right)\label{app:var}
\end{eqnarray}
and this leads to
$$
\widehat{\Sigma}_N=\frac{1}{N}\sum_{t=1}^NQ_t^{-1}e_te_t'=\frac{1}{N}\sum_{t=1}^Nr_te_t',
$$
since
$$
r_t=y_t-m_t'F_t=y_t-m_{t-1}'G_t'F_t-e_tA_t'F_t=(1-A_t'F_t)e_t=Q_t^{-1}(Q_t-F_t'P_{t-1}F_t/\delta)e_t=Q_t^{-1}e_t.
$$
To prove that the second partial derivative of $\ell (\Sigma ;y^N)$
with respect to $\Sigma$ is a negative definite matrix, first we
show that the second partial derivative of $\ell (\Sigma;y^N)$ with
respect to $\Sigma^{-1}$ is a negative definite matrix. Denote with
$D_p$ the duplication matrix (i.e.
$\textrm{vec}(\Sigma)=D_p\textrm{vech}(\Sigma)$, where
$\textrm{vech}(\cdot)$ is the column stacking operator of a lower
portion of a symmetric matrix) and write $H_p$ to be any left
inverse of $D_p$ (i.e. $H_pD_p=I_p$). One choice for $H_p$ is
$H_p=(D_p'D_p)^{-1}D_p'$. For any vector $a$, let $\textrm{diag}(a)$
denote the diagonal matrix with diagonal elements the elements of
$a$. Write $\sigma=\textrm{vech}(\Sigma)$ and
$\sigma_*=\textrm{vech}(\Sigma^{-1})$. From equation (\ref{app:var})
we have
\begin{eqnarray}
\frac{\partial^2 {\ell (\Sigma;y^N)}}{\partial \sigma_*
\partial
\sigma_*'}&=&-NH_p(\Sigma\otimes\Sigma
)D_p-\frac{N}{2}\frac{\partial
\textrm{vech}(\textrm{diag}\{\sigma_{11},\sigma_{22},\ldots,\sigma_{pp}\})}{\partial
\sigma_*'}\nonumber\\
&=&-NH_p(\Sigma\otimes\Sigma )D_p-\frac{N}{2}\frac{\partial
\textrm{vech}(\textrm{diag}\{\sigma_{11},\sigma_{22},\ldots,\sigma_{pp}\})}{\partial
\sigma'} \frac{\partial
\sigma}{\partial \sigma_*'}\nonumber\\
&=&-NH_p(\Sigma\otimes\Sigma )D_p +
\frac{N}{2}\textrm{diag}\{\textrm{vech}(I_p)\}H_p(\Sigma\otimes\Sigma)D_p\nonumber\\
&=&
-\frac{N}{2}[2I_{p(p+1)/2}-\textrm{diag}\{\textrm{vech}(I_p)\}]H_p(\Sigma\otimes\Sigma
)D_p<0,\label{app:var2}
\end{eqnarray}
which is a negative definite matrix, since both
$2I_{p(p+1)/2}-\textrm{diag}\{\textrm{vech}(I_p)\}$ and
$H_p(\Sigma\otimes\Sigma )D_p$ are positive definite.

Now using the chain rule for matrix differentiation we have
$$
\frac{\partial^2 \ell (\Sigma;y^N) }{\partial \sigma\partial
\sigma'}= \frac{\partial^2 \ell (\Sigma;y^N)}{\partial
\sigma_*\sigma_*'}\left(\frac{\partial \sigma_*}{\partial
\sigma'}\right)^2+\frac{\partial \ell (\Sigma;y^N)}{\partial
\sigma_*'}\frac{\partial^2 \sigma_*}{\partial\sigma\partial \sigma'}
$$
and at $\Sigma=\widehat{\Sigma}_N$ we have that
$$
\frac{\partial^2 \ell (\Sigma;y^N) }{\partial \sigma\partial
\sigma'}\bigg |_{\Sigma=\widehat{\Sigma}_N}= \frac{\partial^2 \ell
(\Sigma;y^N)}{\partial \sigma_*\sigma_*'}\bigg
|_{\Sigma=\widehat{\Sigma}_N}\left(\frac{\partial \sigma_*}{\partial
\sigma'}\right)^2\bigg |_{\Sigma=\widehat{\Sigma}_N}<0,
$$
which from (\ref{app:var2}) is a negative definite matrix and so
$\widehat{\Sigma}_N$ maximizes the log-likelihood function $\ell
(\Sigma;y^N)$.
\end{proof}

Before we prove Theorem \ref{th2}, we give the following lemma.

\begin{lem}\label{lemma}
Suppose that the $p\times p$ matrix $B$ follows the singular
multivariate beta distribution $B\sim\B_p(m/2,n/2)$, with density
$$
p(B)= \pi^{(n^2-pn)/2}
\frac{\Gamma_p\{(m+n)/2\}}{\Gamma_n(n/2)\Gamma_p(m/2)}
|K|^{(n-p-1)/2} |B|^{(m-p-1)/2},
$$
where $n$ is a positive integer, $m>p-1$, $I_p-B=H_1KH_1'$, $K$ is
the diagonal matrix with diagonal elements the positive eigenvalues
of $I_p-B$, and $H_1$ is a matrix with orthogonal columns, i.e.
$H_1H_1'=I_p$. For any non-singular matrix $A$, the density of
$X=AB^{-1}A'$, is
$$
p(X)=\pi^{(n^2-pn)/2}
\frac{\Gamma_p\{(m+n)/2\}}{\Gamma_n(n/2)\Gamma_p(m/2)} |A|^{n+m-p-1}
|L|^ {-(p-n+1)/2} |X|^{-(m-p-1)/2},
$$
where $L$ is the diagonal matrix including the positive eigenvalues
of $I_p-A'X^{-1}A$.
\end{lem}
\begin{proof}
First note that $X$ is a non-singular matrix and
$|B|=|A|^2|X|^{-1}$. From D\'{i}az-Garc\'{i}a and Guti\'{e}rrez
(1997), the Jacobian of $B$ with respect to $X$ is
$$
(\,dB)=|K|^{(p-n+1)/2} |L|^{-(p-n+1)/2} |A|^n (\,dX),
$$
where $K$ is defined as in the theorem. Then from the singular
multivariate beta density of $B$ we obtain
\begin{eqnarray*}
p(X)&=&\pi^{(n^2-pn)/2}
\frac{\Gamma_p\{(m+n)/2\}}{\Gamma_n(n/2)\Gamma_p(m/2)} |A|^n
|K|^{(n-p-1)/2} |B|^{(m-p-1)/2} \\ && \times |K|^{(p-n+1)/2}
|L|^{-(p-n+1)/2},
\end{eqnarray*}
from which we immediately get the required density of $X$.
\end{proof}

\begin{proof}[Proof of Theorem \ref{th2}]
First we derive the likelihood function
$L(\Sigma_1,\Sigma_2,\ldots,\Sigma_N;y^N)$. We have
\begin{eqnarray*}
L(\Sigma_1,\Sigma_2,\ldots,\Sigma_N;y^N) &=&
p(y_1,y_2,\ldots,y_N|\Sigma_1,\Sigma_2,\ldots,\Sigma_N) \\ &=&
p(y_N|\Sigma_N,y^{N-1})
p(y_1,y_2,\ldots,y_{N-1}|\Sigma_1,\Sigma_2,\ldots,\Sigma_N).
\end{eqnarray*}
By Bayes' theorem the last part of the right hand side is
$$
p(y_1,y_2,\ldots,y_{N-1}|\Sigma_1,\Sigma_2,\ldots,\Sigma_N)\propto
p(\Sigma_N|\Sigma_{N-1},y^{N-1})
p(y_1,y_2,\ldots,y_{N-1}|\Sigma_1,\Sigma_2,\ldots,\Sigma_{N-1})
$$
and so applying the last equation repeatedly we have
\begin{equation}\label{logl1}
L(\Sigma_1,\Sigma_2,\ldots,\Sigma_N;y^N)=c^*\prod_{t=1}^N
p(y_t|\Sigma_t,y^{t-1})p(\Sigma_t|\Sigma_{t-1},y^{t-1}).
\end{equation}
The density $p(y_t|\Sigma_t,y^{t-1})$ is a multivariate normal
density, since from the Kalman filter
$y_t|\Sigma_t,y^{t-1}\sim\N_{p\times
1}(m_{t-1}'G_t'F_t,Q_t\Sigma_t)$. The density
$p(\Sigma_t|\Sigma_{t-1},y^{t-1})$ is the density $p(X)$ of Lemma
\ref{lemma} with $A=\beta^{1/2}C_{t-1}^{-1}$,
$\Sigma_t=C_t^{-1}(C_t^{-1})'$,
$m=p^{-1}\textrm{tr}(\beta)/(1-p^{-1}\textrm{tr}(\beta))+p-1$ and
$n=1$. The required formula of the log-likelihood function is
obtained from (\ref{logl1}) by taking the logarithm of
$L(\Sigma_1,\Sigma_2,\ldots,\Sigma_N;y^N)$, for $c^*=1$.
\end{proof}


\begin{thebibliography}{}

\bibitem{Aguilar00}
Aguilar, O. and West, M. (2000) Bayesian dynamic factor models and
portfolio allocation. {\it Journal of Business and Economic
Statistics}, {\bf 18}, 338-357.

\bibitem{Asai06}
Asai, M,, McAleer, M. and Yu, J. (2006) Multivariate stochastic
volatility: A review. {\it Econometric Reviews}, {\bf 25}, 145-175.

\bibitem{Audrino06}
Audrino, F. and Barone-Adesi, G. (2006) Average conditional
correlation and tree structures for multivariate GARCH models. {\it
Journal of Forecasting}, {\bf 25}, 579-600.

\bibitem{Basle96}
Basle Committee on Banking Supervision. (1996) Supervisory framework
for the use of ``backtesting'' in conjuction with the internal
models approach to market risk capital requirements.

\bibitem{Basle98}
Basle Committee on Banking Supervision. (1998) Amendment to the
capital accord to incorporate market risk.

\bibitem{Bauwens06}
Bauwens, L., Laurent, S. and Rombouts, J.V.K. (2006) Multivariate
GARCH models: A survey. {\it  Journal of Applied Econometrics}, {\bf
21}, 79-109.

\bibitem{Bollorslev90}
Bollerslev, T. (1990) Modelling the coherence in short-run nominal
exchange rates - a multivariate generalized ARCH model. {\it Review
of Economics and Statistics}, {\bf 72}, 498-505.

\bibitem{Bollerslev88}
Bollerslev, T., Engle, R.F. and Wooldridge, J.M. (1988) A
capital-asset pricing model with time-varying covariances. {\it
Journal of Political Economy}, {\bf 96}, 116-131.

\bibitem{brooks}
Brooks, C. and Persand, G. (2003) Volatility forecasting for risk
management. {\it Journal of Forecasting}, {\bf 22}, 1-22.

\bibitem{Chib02}
Chib, S., Nardari, F. and Shephard, N. (2002) Markov chain Monte
Carlo methods for stochastic volatility models. {\it Journal of
Econometrics}, {\bf{108}}, 281-316.

\bibitem{Chong04}
Chong, J. (2004) Value at Risk from econometric models and implied
from currency options. {\it Journal of Forecasting}, {\bf 23},
603-620.

\bibitem{Comte03}
Comte, F. and Lieberman, O. (2003) Asymptotic theory for
multivariate GARCH processes. {\it Journal of Multivariate
Analysis}, {\bf 84}, 61-84.

\bibitem{Cuaresma}
Cuaresma, J.C. and Hlouskova, J. (2005) Beating the random walk in
central and eastern Europe. {\it Journal of Forecasting}, {\bf 24},
189-201.

\bibitem{De}
De Gooijer, J.G. and Hyndman, R.J. (2006) 25 years of time series
forecasting. {\it International Journal of Forecasting}, {\bf 22},
443-473.

\bibitem{Diaz97}
D\'{i}az-Garc\'{i}a, J.A. and Guti\'{e}rrez, J.R. (1997) Proof of
the conjectures of H. Uhlig on the singular multivariate beta and
the jacobian of a certain matrix transformation. {\it Annals of
Statistics}, {\bf 25}, 2018-2023.

\bibitem{Diebold89}
Diebold, F.X. and Nerlove, M. (1989) The dynamics of exchange-rate
volatility - a multivariate latent factor ARCH model. {\it Journal
of Applied Econometrics}, {\bf 4}, 1-21.

\bibitem{Durbin01}
Durbin, J. and Koopman, S.J. (2001). {\it Time Series Analysis by
State-Space Methods}. Oxford University Press, Oxford.

\bibitem{Engle95}
Engle, R.F. and Kroner, K.F. (1995) Multivariate simultaneous
generalized ARCH. {\it Econometric Theory}, {\bf 11}, 122-150.

\bibitem{Engle90}
Engle, R.F., Ng, V.K. and Rothschild, M. (1990) Asset pricing with
factor-ARCH covariance structure - empirical estimates for treasury
bills. {\it Journal of Econometrics}, {\bf 45}, 213-237.

\bibitem{Fernandez}
Fern\'{a}ndez, E.J. and Harvey, A.C. (1990) Seemingly unrelated time
series equations and a test of homgeneity. {\it Journal of Business
and Economic Statistics}, {\bf 8}, 71-81.

\bibitem{Fiorentini96}
Fiorentini G. and Maravall, A. (1996) Unobserved components in ARCH
models: An application to seasonal adjustment. {\it Journal of
Forecasting}, {\bf 15}, 175-201.

\bibitem{Gupta}
Gupta, A.K. and Nagar, D.K. (1999) {\it Matrix Variate
Distributions}. Chapman and Hall, New York.

\bibitem{Harrison87}
Harrison, P.J. and West, M. (1987) Practical Bayesian forecasting.
{\it The Statistician}, {\bf 36}, 115-125.

\bibitem{Harvey86}
Harvey, A.C. (1986) Analysis and generalisation of a multivariate
exponential smoothing model. {\it Management Science}, {\bf{32}},
374-380.

\bibitem{Harvey89}
Harvey, A.C. (1989) {\it Forecasting Structural Time Series
Models and the Kalman Filter}. Cambridge University Press,
Cambridge.

\bibitem{HarveyS}
Harvey, A.C. and Snyder, R.D. (1990) Structural time series models
in inventory control. {\it International Journal of Forecasting},
{\bf 6}, 187-198.

\bibitem{Harvey94}
Harvey, A.C., Ruiz, E. and Shephard, N. (1994) Multivariate
stochastic variance models. {\it Review of Economic Studies},
{\bf{61}}, 247-264.

\bibitem{Harville97}
Harville, D.A. (1997) {\it Matrix Algebra from a Statistician's
Perspective}. Springer-Verlag, New York.

\bibitem{Jacquier94}
Jacquier, E., Polson, N.G. and Rossi, P.E. (1994) Bayesian analysis
of stochastic volatility models (with discussion). {\it Journal of
Business and Economic Statistics}, {\bf{12}}, 371-419.

\bibitem{Kim98}
Kim, S., Shephard, N. and Chib, S. (1998) Stochastic volatility:
Likelihood inference and comparison with ARCH models. {\it Review of
Economic Studies}, {\bf 65}, 361-393.

\bibitem{Liensfield}
Liesenfeld, R. and Richard, J.F. (2006) Classical and Bayesian
analysis of univariate and multivariate stochastic volatility
models. {\it Econometric Reviews}, {\bf 25}, 335-360.

\bibitem{Lutkepohl93}
L\"utkepohl, H. (1993) {\it Introduction to Multiple Time Series
Analysis}. Springer-Verlag, Berlin.

\bibitem{Maasoumi06}
Maasoumi, E. and McAleer, M. (2006) Multivariate stochastic
volatility: An overview. {\it Econometric Reviews}, {\bf 25},
139-144.

\bibitem{McKenzie01}
McKenzie, M. Michell, H. Brooks, R.D. and Faff, R.W. (2001) Power
ARCH modelling of commodity futures data on the London Metal
Exchange. {\it European Journal of Finance}, {\bf 7}, 22-38.

\bibitem{Meyer03}
Meyer, R., Fournier, D. and Berg, A. (2003) Stochastic volatility:
Bayesian computation using automatic differentiation and the
extended Kalman filter. {\it The Econometrics Journal}, {\bf
6},408-420.

\bibitem{Morgan96}
Morgan, J.P. (1996) {\it RiskMetrics Technical Document}, 4th edn,
New York.

\bibitem{Philipov2}
Philipov, A. and Glickman, M.E. (2006a) Multivariate stochastic
volatility via Wishart processes. {\it Journal of Business and
Economic Statistics}, {\bf 24}, 313-328.

\bibitem{Philipov1}
Philipov, A. and Glickman, M.E. (2006b) Factor multivariate
stochastic volatility via Wishart processes. {\it Econometric
Reviews}, {\bf 25}, 311-334.

\bibitem{Pitt99}
Pitt M.K. and Shephard, N. (1999) Filtering via simulation:
Auxiliary particle filters. {\it Journal of the American Statistical
Association}, {\bf 94}, 590-599.

\bibitem{Quintana87}
Quintana, J.M. and West, M. (1987). An analysis of international
exchange rates using multivariate DLMs. {\it The Statistician}, {\bf
36}, 275-281.

\bibitem{Salvador1}
Salvador, M. and Gargallo, P. (2004). Automatic monitoring and
intervention in multivariate dynamic linear models. {\it
Computational Statistics and Data Analysis}, {\bf 47}, 401-431.

\bibitem{Salvador2}
Salvador, M. and Gargallo, P. (2005). Automatic selective
intervention in dynamic linear models . {\it Journal of Applied
Statistics}, {\bf 30}, 1161-1184.

\bibitem{Salvador3}
Salvador, M. and Gargallo, P. (2006). Automatic detection and
identification of shocks in Gaussian state-space models: A Bayesian
approach. {\it Applied Stochastic Models in Business and Industry}
{\bf 22}, 17-39.

\bibitem{Salvador4}
Salvador, M., Gallizo, J.L. and Gargallo, P. (2003). A dynamic
principal components analysis based on multivariate matrix normal
dynamic linear models. {\it Journal of Forecasting}, {\bf 22},
457-478.

\bibitem{Salvador5}
Salvador, M., Gallizo, J.L. and Gargallo, P. (2004). Bayesian
inference in a matrix normal dynamic linear model with unknown
covariance matrices. {\it Statistics}, {\bf 38}, 307-335.

\bibitem{Shephard93}
Shephard, N. (1993) Fitting nonlinear time series models with
applications to stochastic variance models. {\it Journal of Applied
Econometrics}, {\bf 8}, 135-152.

\bibitem{Shephard97}
Shephard, N. and Pitt, M.K. (1997) Likelihood analysis of
non-Gaussian measurement time series. {\it Biometrika} {\bf 84},
653-667.

\bibitem{Srivastava03}
Srivastava, M.S. (2003) Singular Wishart and multivariate beta
distributions. {\it Annals of Statistics}, {\bf 31}, 1537-1560.

\bibitem{Triantafyllopoulos07}
Triantafyllopoulos, K. (2007) Feedback quality adjustment with
Bayesian state space models. {\it Applied Stochastic Models in
Business and Industry}, (to appear).

\bibitem{Triantafyllopoulos06a}
Triantafyllopoulos, K. (2006a) Multivariate discount weighted
regression and local level models. {\it Computational Statistics and
Data Analysis}, {\bf 50}, 3702-3720.

\bibitem{triantafyllopoulos06b}
Triantafyllopoulos, K. (2006b) Multivariate control charts based on
Bayesian state space models. {\it Quality and Reliability
Engineering International}, {\bf 22}, 693-707.

\bibitem{Triantafyllopoulos02}
Triantafyllopoulos, K. and Pikoulas, J. (2002). Multivariate
Bayesian regression applied to the problem of network security. {\it
Journal of Forecasting}, {\bf 21}, 579-594.

\bibitem{Tsay02}
Tsay, R.S. (2002). {\it Analysis of Financial Time Series}. Wiley,
New York.

\bibitem{Tse}
Tse, Y.K. and Tsui, A.K.C. (2002) A multivariate generalized
autoregressive conditional heteroscedasticity model with
time-varying correlations. {\it Journal of Business and Economic
Statistics}, {\bf 20}, 351-362.

\bibitem{Uhlig94}
Uhlig, H. (1994) On singular Wishart and singular multivariate beta
distributions. {\it Annals of Statistics}, {\bf{22}}, 395-405.

\bibitem{Uhlig97}
Uhlig, H. (1997) Bayesian vector autoregressions with stochastic
volatility. {\it Econometrica}, {\bf{65}}, 59-73.

\bibitem{Watkins04}
Watkins, C. and McAleer, M. (2004). Econometric modelling of
non-ferrous metal prices. {\it Journal of Economic Surveys}, {\bf
18}, 651-701.

\bibitem{West97}
West, M. and Harrison, P.J. (1997). {\it Bayesian Forecasting and
Dynamic Models}. Springer-Verlag, 2nd edn., New York.

\bibitem{Wong}
Wong, H. and Li, W.K. (1997) On a multivariate conditional
heteroscedastic model. {\it Biometrika}, {\bf 84}, 111-123.

\bibitem{Yu}
Yu, J. and Meyer, R. (2006) Multivariate stochastic volatility
models: Bayesian estimation and model comparison. {\it Econometric
Reviews}, {\bf 25}, 361-384.

\end{thebibliography}
\end{document}